\documentclass[letterpaper, 10pt, conference]{ieeeconf}      

\IEEEoverridecommandlockouts                              

\usepackage{amsmath,amssymb,amsthm} 
\usepackage[T1]{fontenc}
\usepackage{microtype}  
\usepackage{tabularx} 
\usepackage{dsfont}
\usepackage{xcolor}

\newtheorem{theorem}{Theorem}[section]

\newtheorem{lemma}[theorem]{Lemma}

\newtheorem{remark}{Remark}

\usepackage{graphicx} 
\usepackage{cite}
\usepackage{booktabs}
\usepackage{algorithm}
\usepackage{algorithmic}

\usepackage{tikz}
\usetikzlibrary{shapes, arrows, positioning, calc, chains}

\makeatletter
\DeclareRobustCommand{\rvdots}{%
	\vbox{
		\baselineskip4\p@\lineskiplimit\z@
		\kern-\p@
		\hbox{.}\hbox{.}\hbox{.}
}}
\makeatother

\makeatletter
\def\footnoterule{\relax%
	\kern-5pt
	\hbox to \columnwidth{\hfill\vrule width .9\columnwidth height 0.4pt\hfill}
	\kern4.6pt}
\makeatother

\usepackage{hyperref}
\definecolor{darkblue}{rgb}{0.0,0.0,0.6}
\hypersetup{colorlinks,breaklinks,linkcolor=darkblue,urlcolor=darkblue,anchorcolor=darkblue,citecolor=darkblue}

\usepackage{subcaption}

\makeatother
\title{\huge Dynamic Weight Optimization for Double Linear Policy: A Stochastic Model Predictive Control Approach}

\author{\large Tan Chin Hong,$^{*}$ and Chung-Han Hsieh${}^{\dagger}$ \emph{Member, IEEE}
	\thanks{
    This paper is partially supported by the National Science and Technology Council~(NSTC), Taiwan, under Grant: NSTC114--2628--E--007--006--. ${}^{\dagger}$Corresponding Author: Chung-Han Hsieh is with the Department of Quantitative Finance, National Tsing Hua University, Hsinchu 300044, Taiwan. E-mail: \href{mailto: ch.hsieh@mx.nthu.edu.tw}{ch.hsieh@mx.nthu.edu.tw}. 
    } 
	\thanks{
   ${}^{*}$Tan Chin Hong is with the Institute of Statistics and Data Science, National Tsing Hua University, Hsinchu 300044, Taiwan. E-mail: \href{mailto: tan.family0312@gmail.com}{tan.family0312@gmail.com}.
    }
}

\begin{document}
	
\maketitle
\thispagestyle{empty}
\pagestyle{empty}
	
	
\begin{abstract} 
    The Double Linear Policy (DLP) framework guarantees a Robust Positive Expectation (RPE) under optimized constant-weight designs or admissible prespecified time-varying policies. However, the sequential optimization of these time-varying weights remains an open challenge. To address this gap, we propose a Stochastic Model Predictive Control (SMPC) framework. We formulate weight selection as a receding-horizon optimal control problem that explicitly maximizes risk-adjusted returns while enforcing survivability and predicted positive expectation constraints. Notably, an analytical gradient is derived for the non-convex objective function, enabling efficient optimization via the L-BFGS-B algorithm. 
    Empirical results demonstrate that this dynamic, closed-loop approach improves risk-adjusted performance and drawdown control relative to constant-weight and prescribed time-varying DLP baselines. 
\end{abstract}

	

\section{Introduction}
The \emph{Simultaneous Long-Short (SLS)} trading controller, pioneered in \cite{Barmish2011OnAP, Barmish2016OnAN, Barmish2011OnPL}, introduced the use of linear feedback control in robust algorithmic trading; see also \cite{barmish2024jump} for a recent tutorial. The defining feature of this paradigm is the \emph{Robust Positive Expectation} (RPE) property, which guarantees positive expected cumulative gain-loss across a wide class of asset-price processes. This theoretical property has motivated numerous extensions, including modifications for delay, cross-coupling, and Proportional-Integral (PI) control; see, for instance,~\cite{Abbracciavento2023CrossCoupledSF, Chen2025OnRO, deshpande2020simultaneous, Malekpour2018AGO, Deshpande2018AGO}.

Building on this foundation, the \emph{Double Linear Policy} (DLP),~\cite{hsieh2022robust} modified the SLS strategy to enable optimal weight selection via mean-variance criteria.
Early studies largely focused on constant-weight designs, including extensions for transaction costs \cite{Hsieh_2023}. Subsequent work by~\cite{wang2023robustness} established that the survivability and RPE properties are preserved under time-varying weights; however, those weight functions are specified rather than generated by an optimization principle.  
More recently, \cite{hsieh2025robust} extended DLP to more general multi-asset lattice markets.

Despite these advances, the question of how to \emph{optimally} select weights in a dynamic environment remains an open challenge. Prior work has largely relied on heuristic schemes or backward-looking statistical calibration---such as finding optimal constant weights over a historical window \cite{hsieh2025robust}. These approaches are fundamentally retrospective and may fail to adapt in real time to time-varying market conditions. To address this, Model Predictive Control (MPC) offers a promising alternative. While MPC has been widely applied in quantitative finance \cite{primbs2018applications, MultiObjective_melo, Primbs03062018, herzog2007stochastic}, its application to the specific multiplicative geometry and survivability constraints of DLP remains largely~unexplored.

In this paper, we propose a Stochastic MPC approach for the DLP framework. Unlike the calibration-based methods in~\cite{hsieh2025robust}, our SMPC approach generates a {dynamic sequence} of weights through receding-horizon optimization. The controller explicitly maximizes risk-adjusted returns while adhering to survivability and a \emph{predicted positive expectation} constraint.
Notably, we derive an analytical gradient for the non-convex objective, circumventing finite-difference approximations and enabling efficient numerical solution via the classical Limited-memory Broyden–Fletcher–Goldfarb–Shanno with Box constraints (L-BFGS-B) algorithm, e.g., see \cite{byrd1995limited, bemporad2025bfgs}, an extension of classical L-BFGS for handling bound constraints; see \cite{nocedal2006numerical}.


\section{Preliminaries}\label{section: preliminaries}

\subsection{Market and Account Dynamics}
Let $(\Omega,\mathcal{F}, \mathbb{P})$ be a complete probability space equipped with the filtration $\{\mathcal{F}_k\}_{k \in \mathbb{N}_0}$, where $\mathcal{F}_0 := \{\emptyset, \Omega\}$ up to $\mathbb{P}$-null sets, and $\mathcal{F}_k := \sigma(\{X(0), \dots, X(k-1)\})$ represents the information available up to time $k \geq 1$. Let $S(0) > 0$, and for $k=1,2,\dots$, let $S(k) > 0$ denote the risky asset price. The corresponding per-period return at time $k$ is defined as~$X(k) := \tfrac{S(k+1) - S(k)}{S(k)}.$ We assume that $X(k) \in [X_{\min}, X_{\max}]$ almost surely, where the deterministic bounds satisfy $-1 < X_{\min} < 0 < X_{\max} < \infty$.
Furthermore, for each time~$k$, we assume that  the future return sequence $\{X(i)\}_{i \geq k}$ is \emph{conditionally independent} given $\mathcal{F}_k$, in the sense that for any finite horizon~$H \geq 1$, the collection $\{X(i)\}_{i=k}^{k+H-1}$ is mutually independent given $\mathcal{F}_k$. 
For all $i \geq k$, we define the conditional mean $\mathbb{E}_k[X(i)] := \mathbb{E}[X(i) \mid \mathcal{F}_k] =: \mu_i^{(k)}$  and conditional variance $\operatorname{var}_k(X(i)) :=\mathbb{E}_k [\left( X(i) - \mathbb{E}_k[X(i)]\right)^2 ] =:  (\sigma_i^{(k)})^2 \geq 0$, which are $\mathcal{F}_k$-measurable.

\subsection{The Double Linear Policy and Account Value Dynamics}\label{DLP Formulation}
Following the standard DLP setting \cite{Hsieh_2023, wang2023robustness, hsieh2025robust}, the initial account $V(0) := V_0 >0$ is partitioned by $\alpha \in [0, 1]$ into two accounts: $V_L(0):= \alpha V_0$ and $V_S(0) = (1-\alpha)V_0$. The trading policy $\pi(k)$ at time $k$ is~$\pi(k):= \pi_L(k) + \pi_S(k)$, where the long and short components $\pi_L$ and $\pi_S$  follow the double linear~form:
\begin{align} \label{eq: double linear policy}
    \begin{cases}
        \pi_L(k) = w_L(k) V_L(k); \\
        \pi_S(k) = -w_S(k) V_S(k).
    \end{cases} 
\end{align}
where $w_L(k), w_S(k)$ are weighting functions satisfying
$
    (w_L(k), w_S(k)) \in \mathcal{W},
$
and the feasible set $\mathcal{W}$ satisfies
\[
\mathcal{W}  := \left\{ (w_L(k), w_S(k)) :  0\le w_L(k), w_S(k)\le w_{\max} \right\}
\]
with $w_{\max}  := \min \left\{ 1, \tfrac{1}{X_{\max} }  \right\}> 0$. 
The account values under~$\pi_L(k)$ and $\pi_S(k)$, denoted by~$V_L(k)$ and $V_S(k)$, is described by the following stochastic recursive equations:
\begin{align} \label{eq: DLP account dynamics}
    \begin{cases}
        V_L(k+1) = V_L(k) + X(k)\pi_L(k); \\ 
        V_S(k+1) = V_S(k) + X(k)\pi_S(k),
    \end{cases}
\end{align}
and the total account value is given by $V(k):=V_L(k) + V_S(k)$.

\subsection{Survivability and Robust Positive Expectation~(RPE)} \label{RPE}
We introduce two desirable properties that a trading policy should satisfy: First, the policy must have \emph{survivability} if~$V(k) >0$ for all~$k$ with probability one. 
By the definition of the admissible set $\mathcal{W}$, bounding $w_L(k) \leq 1$ and $w_S(k) \leq 1/X_{\max}$ inherently guarantees this property, as both the long and short multipliers governing the account dynamics in \eqref{eq: DLP account dynamics} are bounded strictly away from zero \cite{wang2023robustness}.

Second, the policy satisfies the~\emph{Robust Positive Expectation} (RPE) property if, under all market conditions, the expected cumulative gain is non-negative:
$
    \mathbb{E}[V(k) - V_0] \geq 0
$  
for all~$k$.

\section{Problem Formulation}\label{section: smpc formulation}
 Prior work~\cite{wang2023robustness} established the Robust Positive Expectation (RPE) property for the DLP  under a different return model, namely, independent returns with common mean and common variance.
In contrast, we instead adopt a more general modeling assumption that the future returns  $\{X(i)\}_{i \geq k}$ are conditionally independent given $\mathcal{F}_k$. Under this model, we formulate a stochastic model predictive control problem to select the common weight sequence $\{w_k\}_{k\geq 0}$ that maximizes the risk-adjusted terminal wealth, while enforcing survivability and a \emph{predicted positive expectation} property.

\subsection{State-Space Model}
We enforce a symmetric weighting scheme for long and short positions, i.e., $w_k:= w_L(k) = w_S(k)$, and we set $\alpha = 1/2$. Consider a system where the state vector $\mathbf{z}_k \in \mathbb{R}^2$ consists of the long and short account values, $V_L(k)$ and $V_S(k)$, and the scalar output ${y}_k=V(k)$ represents the total account value.  We define
\begin{align}
    \mathbf{z}_k 
        &:= \begin{bmatrix} 
            V_L(k) \\ 
            V_S(k) 
            \end{bmatrix}  \in \mathbb{R}^{2},
            \\ 
    y_k &:=  \mathbf{c}^\top\, \mathbf{z}_k = V(k) \in \mathbb{R}\label{eq: output}, 
\end{align}
where $\mathbf{c}:=[1 \; 1]^\top$. The initial state is given by
$
\mathbf{z}_0  :=  
    \begin{bmatrix} 
        V_L(0) & V_S(0) 
    \end{bmatrix}^\top$ 
with $y_0 = V_0 >0$. 
Recalling the individual long/short account dynamics from~\eqref{eq: DLP account dynamics}, the state evolution is governed by the time-varying linear stochastic system:
\begin{align}\label{eq: linear system}
    \mathbf{z}_{k+1} = A_k(w_k, X(k)) \, \mathbf{z}_k
\end{align}
where the transition matrix~$A_k(w_k, X(k)) \in \mathbb{R}^{2\times 2}$, denoted by $A_k$ for brevity, is defined at each time step $k$ as:
$$
A_k := 
    \begin{bmatrix} 
        1 + X(k) w_k & 0 \\ 
        0 & 1 - X(k) w_k 
    \end{bmatrix}.
$$ 
Under the time-varying system~\eqref{eq: linear system} and the output map~\eqref{eq: output}, the evolved total account value at horizon $H >0$ is given by
\begin{align*}
   y_{k+H} =  \mathbf{c}^\top\Phi_{k+H, k} \, \mathbf{z}_k
\end{align*}
where $\Phi_{k+H, k}$ is the state transition matrix, defined via the left-ordered product $\Phi_{k+H, k}:=  A_{k+H-1} A_{k+H-2} \cdots A_k.$

\subsection{Stochastic MPC with Mean-Variance Objective} \label{subsection: SMPC with MV objective}
For a prediction horizon $H > 0$, the SMPC problem maximizes the \emph{risk-adjusted predicted wealth} at the end of the horizon over the control sequence $\{w_i\}_{i=k}^{k+H-1} \subseteq \mathcal{W}$, subject to constraints on weights, predicted survivability, and a predicted positive expectation. 
\begin{align}
    \max_{\{w_i\}_{i=k}^{k+H-1}} \quad & \mathbb{E}_k[y_{k+H}] - \gamma \operatorname{var}_k(y_{k+H}) \label{eq:mv-problem} \\
    \text{s.t.} \quad 
    & 0 \leq w_i \leq w_{\max}, \quad i =k, \dots, k+H-1\\
    & y_i \geq 0 \; {\rm a.s.}, \quad i =k+1, \dots, k+H\label{eq:survivability}\\
    & \mathbb{E}_k[ y_{k+H} -   y_k] \geq 0. \label{eq:ppe}
\end{align}

\begin{remark}[RPE versus Predicted Positive Expectation] \rm
    While the standard DLP guarantees an unconditional robust positive expectation (RPE), the constraint~\eqref{eq:ppe} enforces a localized $H$-step-ahead \emph{predicted positive expectation} (PPE). When $H = 1$, the constraint explicitly forces the system to be a submartingale, i.e., $\mathbb{E}_k[y_{k+1}] \geq y_k$. Consequently, by invoking the tower property of conditional expectation, this single-step PPE condition inherently implies RPE. For~$H > 1$, this strict single-step guarantee is relaxed, and the constraint instead acts as a structural regularizer to ensure multi-step positive expected~growth.
\end{remark}

\subsection{Survivability Considerations}
The following lemma establishes a \emph{trajectory-wide predicted survivability} property for the proposed SMPC framework, guaranteeing that the future account value remains strictly positive at every step in the prediction horizon whenever~$0 \leq w_i \leq w_{\max}$.

\begin{lemma}[Trajectory-Wide Predicted Survivability] \label{lemma: predicted survivability}
    Fix a prediction horizon $H>0$. Suppose the current account values satisfy $V_L(k) >0$ and $V_S(k) \geq 0$. If the weight sequence satisfy $0 \leq w_i \leq w_{\max}$ for all $i \in \{k, k+1, \dots, k+H-1\}$, then for all $h=1,\dots,H$, the system output satisfies
    $$
        y_{k+h} >0 \, \text{ a.s.}
    $$
\end{lemma}
\begin{proof}
    Recall that $y_{k+h} =V_L(k+h)+V_S(k+h)$ for all $h=1,\dots,H$. Since $w_i \in [0, w_{\max}]$ and $X(i) \in [X_{\min}, X_{\max}]$ a.s., the long account evolves as
    \begin{align*}
        V_L(k+h) 
        &= V_L(k)\prod_{i=k}^{k+h-1}(1 + w_iX(i)) \\ 
        &\geq V_L(k)(1 + w_{\max}X_{\min})^h > 0,
    \end{align*}
    where strict positivity follows from $w_{\max}\leq 1$, $X_{\min}>-1$, and $V_L(k)>0$. 
    Similarly, $w_{\max} \leq \tfrac{1}{X_{\max}}$ and $V_S(k) \geq 0$ yield 
    $
    V_S(k+h)\geq V_S(k)(1 - w_{\max}X_{\max})^h \geq 0.
    $
    The sum of a strictly positive quantity and a non-negative quantity is strictly positive. This holds for any arbitrary $h \in \{1,\dots, H\}.$
\end{proof}

\begin{remark} \rm
    By Lemma~\ref{lemma: predicted survivability}, $y_i > 0$ a.s. is intrinsically guaranteed for the entire trajectory $\{y_{k+1}, \dots, y_{k+H}\}$ for any admissible control sequence $0 \leq w_i \leq w_{\max}$. Consequently, the entire block of survivability constraints~\eqref{eq:survivability} can be dropped without altering the feasible set, reducing the SMPC problem~to:
\begin{align}
    \max_{\{w_i\}_{i=k}^{k+H-1}} \quad & \mathbb{E}_k[y_{k+H}] - \gamma \operatorname{var}_k(y_{k+H}) \label{eq:mv-problem_reduced} \\
    \text{s.t.} \quad 
    & 0 \le w_i \le w_{\max}, \quad  i \in \{k, \dots, k+H-1\}\nonumber\\
    & \mathbb{E}_k[y_{k+H}  - y_k] \geq 0 \label{eq:ppe2}.
\end{align}
 \end{remark}

To explicitly evaluate the objective in \eqref{eq:mv-problem_reduced}, we now derive the analytical expressions for the conditional moments of the predicted account trajectory.

\begin{lemma}[Conditional Moments of Predicted Wealth]\label{lemma: conditional moments}
    Given a prediction horizon $H > 0$, an initial state~$\mathbf{z}_k$, and weight sequence $\{w_i\}_{i=k}^{k+H-1}$, let $
    \Phi_{k+H, k} :=  A_{k+H-1} A_{k+H-2} \cdots A_k
    $ denote the state transition matrix, and define $\mathbf{c}:=[1\; 1]^\top$.
    Then, the conditional expectation and conditional variance of the output $y_{k+H}$ are given by
    \begin{align}
        \mathbb{E}_k[y_{k+H}] 
            &= \mathbf{c}^\top\bar{\Phi}\,\mathbf{z}_k; \label{eq: cond_mean} \\
        \operatorname{var}_k(y_{k+H}) 
        &= \mathbf{z}_k^\top\Sigma\,\mathbf{z}_k,\label{eq: cond_var}
    \end{align}
where $\bar{\Phi}$ is the expected transition matrix defined as 
$$
\bar{\Phi} 
    := \mathbb{E}_k[\Phi_{k+H,k}] = \operatorname{diag}(P^+,\,P^-) \in \mathbb{R}^{2 \times 2}
$$
with $P^\pm:=\prod_{i=k}^{k+H-1}\left(1\pm w_i\mu_i^{(k)}\right)$, and $\Sigma$ is the covariance~matrix 
$$
    \Sigma := M - \bar{\Phi}^\top\,\mathbf{c}\mathbf{c}^\top\bar{\Phi} \in \mathbb{R}^{2\times 2}.
$$
The second-moment matrix $M := \mathbb{E}_k[\Phi_{k+H,k}^\top\mathbf{c}\mathbf{c}^\top\Phi_{k+H,k}]$ 
is symmetric with entries $M_{11} = Q^+$, $M_{22} = Q^-$, and
\[
    M_{12} = M_{21} = \prod_{i=k}^{k+H-1}\!\left(1-w_i^2\left(\left(\mu_i^{(k)}\right)^2+\left(\sigma_i^{(k)}\right)^2\right)\right),
\]
with 
$
Q^\pm:=\prod_{i=k}^{k+H-1}\left((1\pm w_i\mu_i^{(k)})^2+\left(\sigma_i^{(k)}\right)^2 w_i^2\right)
$.

\end{lemma}
\begin{proof}
    The proof follows from the conditional independence of the return sequence and the property of the variance operator. The full derivation is deferred to the~\nameref{appendix}.
\end{proof}

\section{Solving the Stochastic MPC Problem}\label{section: solving the SMPC}
As established in Lemma~\ref{lemma: conditional moments}, the multiplicative dependence of the state transition matrices on the control sequence $\mathbf{w}$ renders the mean-variance objective function highly non-convex. To this end, we employ L-BFGS-B~\cite{byrd1995limited}, a limited-memory quasi-Newton method that approximates second-order curvature, while enforcing the box constraint $0 \leq w_i \leq w_{\max}$.

\subsection{Analytical Gradient Derivation}
The computational bottleneck in applying quasi-Newton methods to nonlinear receding-horizon problems is typically the evaluation of the gradient. 
Crucially, Theorem~\ref{theorem: gradient} below provides the exact analytical gradient~$\nabla_{\mathbf{w}}J(\mathbf{w})$. This closed-form expression avoids finite-difference approximations and enables exact, computationally efficient gradient evaluations at each iteration.

\begin{theorem}[Analytical Gradient of the Objective Function] \label{theorem: gradient}
For $H>0$, let $J(\mathbf{w}) := \mathbb{E}_k[y_{k+H}] - \gamma \operatorname{var}_k(y_{k+H})$ and 
$D:=\operatorname{diag}(1, -1)$. For each $j=0,\dots, H-1$, define 
$\bar{\Phi}_j:= \operatorname{diag}(A_j^+, A_j^-) \in \mathbb{R}^{2\times 2}$ 
where $A_j^\pm:=\prod_{i\in\widetilde{I}_j}\left(1\pm\mu_i^{(k)} w_i\right)$ and 
$\widetilde{I}_j:=\{k,\dots,k+H-1\}\setminus\{k+j\}$. With the matrix $M$ defined in Lemma~\ref{lemma: conditional moments}, define
\[
    {M}'_j := \frac{\partial M}{\partial w_{k+j}} = \begin{bmatrix} {m'}^+_j & {m'}^{LS}_j \\ {m'}^{LS}_j & {m'}^-_j \end{bmatrix}\in \mathbb{R}^{2\times 2},
\]
 whose entries are
\begin{align*}
    {m'}^+_j &:= 2\left(\mu_{k+j}^{(k)}\left(1+\mu_{k+j}^{(k)} w_{k+j}\right)+\left(\sigma_{k+j}^{(k)}\right)^2 w_{k+j}\right)B_j^+,\\
    {m'}^-_j &:= 2\left(-\mu_{k+j}^{(k)}\left(1-\mu_{k+j}^{(k)} w_{k+j}\right)+\left(\sigma_{k+j}^{(k)}\right)^2 w_{k+j}\right)B_j^-,\\
    {m'}^{LS}_j &:= -2\left(\left(\mu_{k+j}^{(k)}\right)^2+\left(\sigma_{k+j}^{(k)}\right)^2\right)w_{k+j}C_j,
\end{align*}
where 
    $B_j^\pm
    :=\prod_{i\in\widetilde{I}_j}\bigl((1\pm\mu_i^{(k)} w_i)^2+ \left(\sigma_i^{(k)}\right)^2 w_i^2\bigr)$ 
and 
    $
    C_j
    :=\prod_{i\in\widetilde{I}_j}\left(1-\left(\left(\mu_i^{(k)}\right)^2+\left(\sigma_i^{(k)}\right)^2\right)w_i^2\right)
    $. 
Then, for $j=0,\dots, H-1$, the analytical partial derivative of $J(\mathbf{w})$ with respect to $w_{k+j}$ is given by
\begin{align*}
    \frac{\partial }{\partial w_{k+j}} J(\mathbf{w})
    &= \mu_{k+j}^{(k)}\,\mathbf{c}^\top D\bar{\Phi}_j\,\mathbf{z}_k 
       - \gamma\,\mathbf{z}_k^\top\Sigma'_j\,\mathbf{z}_k,   
\end{align*}
where $\Sigma'_j := {M}'_j - \mu_{k+j}^{(k)}\bigl(D\bar{\Phi}_j\,\mathbf{c}\mathbf{c}^\top\bar{\Phi} 
+ \bar{\Phi}\,\mathbf{c}\mathbf{c}^\top D\bar{\Phi}_j\bigr)\in\mathbb{R}^{2\times 2}.$
\end{theorem}

\begin{proof}
    The derivation requires applying the product rule to the conditional moments established in Lemma~\ref{lemma: conditional moments}; see the~\nameref{appendix} for the complete derivation.
\end{proof}

\subsection{Handling the Positive Expectation Constraint}
While L-BFGS-B natively handles the box constraints $0 \leq w_i \leq w_{\max}$, the nonlinear positive expectation constraint~\eqref{eq:ppe2} requires an outer penalty framework.
To address this, we employ an Augmented Lagrangian (AL) method~\cite{nocedal2006numerical}. At each time step~$k$, the constrained maximization problem is converted to a sequence of box-constrained subproblems by augmenting the objective with a quadratic penalty and a dual variable $\lambda \geq 0$. Specifically, we define the AL objective as
 {\small
 \[
J_{\rho,\lambda}(\mathbf{w}) = \mathbb{E}_k[y_{k+H}] - \gamma\operatorname{var}_k(y_{k+H}) {\color{blue} -} \frac{\rho}{2}\max\!\left(0,\,-h(\mathbf{w})+\frac{\lambda}{\rho}\right)^2,
\]
}where $h(\mathbf{w}):=\mathbb{E}_k[y_{k+H} - y_k]$ represents the expected gain. 

Each augmented subproblem is solved via \mbox{L-BFGS-B}, after which the dual variable is updated as $\lambda\leftarrow\max(0,\,\lambda-\rho\, h(\mathbf{w}^\star))$, and the penalty parameter $\rho$ is doubled if the constraint violation exceeds a predefined tolerance. Since~$\nabla_\mathbf{w} h(\mathbf{w})=\nabla_\mathbf{w}\mathbb{E}_k[y_{k+H}]$ is available in closed form; see Theorem~\ref{theorem: gradient}, exact analytical gradient evaluation is preserved throughout the optimization.

By approximating second-order curvature, L-BFGS-B typically achieves a faster convergence rate than standard first-order projected-gradient methods, making it well-suited for real-time receding-horizon execution. The complete procedure is outlined in Algorithms~\ref{alg:lbfgsb} and~\ref{alg:stochastic-mpc}.

\begin{algorithm}[htbp]
\small
\caption{L-BFGS-B, \cite{byrd1995limited}}
\label{alg:lbfgsb}
\textbf{Input:} $J$, $\nabla J$, $w_{\max}$, $H$, $\mathbf{w}^{(0)}$, $m$, $\varepsilon$, $K_{\max}$.\\
\textbf{Output:} $\mathbf{w}^\star$.
\begin{algorithmic}[1]
\STATE Initialize $\mathbf{w}\leftarrow\mathbf{w}^{(0)}$,
       $(\mathcal{S},\mathcal{Y})\leftarrow\emptyset$,
       $H_0=I$, $k\leftarrow 0$.
\REPEAT
    \STATE $\mathbf{g}\leftarrow\nabla J(\mathbf{w})$.
    \STATE \textbf{Generalized Cauchy Point.}
           Trace the piecewise-linear path
           $\mathbf{x}(t)=\Pi_{[0,w_{\max}]^{H}}(\mathbf{w}-t\,\mathbf{g})$;
           sort breakpoints $\{t_i\}$ where coordinate $i$ hits a bound, and advance along each segment, minimizing the quadratic model
           $$q(\mathbf{x})=J(\mathbf{w})+\mathbf{g}^\top(\mathbf{x}-\mathbf{w})
           +\tfrac{1}{2}(\mathbf{x}-\mathbf{w})^\top B_k(\mathbf{x}-\mathbf{w})$$
           using the compact L-BFGS representation of $B_k$,
           until $q$ increases or all breakpoints are exhausted.
           Denote the result~$\mathbf{w}^c$.
    \STATE Partition indices:
           $\mathcal{A}_k=\{i: w^c_i\in\{0,w_{\max}\}\}$,\quad
           $\mathcal{B}_k=\{1,\ldots,H\}\setminus\mathcal{A}_k$.
    \IF{$\|\mathbf{g}_{_{\mathcal{B}_k}}\|_\infty \le \varepsilon$}
        \STATE \textbf{break}
    \ENDIF
    \STATE \textbf{Subspace Minimization.}
           Starting from $\mathbf{w}^c$, minimize $q$ over $\mathcal{B}_k$
           with $\mathcal{A}_k$ coordinates fixed,
           via the reduced compact L-BFGS system;
           project onto $[0,w_{\max}]^H$ to obtain $\mathbf{w}^+$.
    \STATE $\mathbf{s}_k\leftarrow\mathbf{w}^+-\mathbf{w}$,\quad
           $\mathbf{y}_k\leftarrow\nabla J(\mathbf{w}^+)-\mathbf{g}$.
    \IF{$\mathbf{s}_k^\top\mathbf{y}_k>0$}
        \STATE Append $(\mathbf{s}_k,\mathbf{y}_k)$ to $(\mathcal{S},\mathcal{Y})$;
               evict oldest pair if $|\mathcal{S}|>m$.
        \STATE $H_0^{(k+1)}\leftarrow
               \tfrac{\mathbf{s}_k^\top\mathbf{y}_k}{\|\mathbf{y}_k\|^2}\,I$.
    \ENDIF
    \STATE $\mathbf{w}\leftarrow\mathbf{w}^+$,\quad $k\leftarrow k+1$.
\UNTIL{$k \ge K_{\max}$}
\RETURN $\mathbf{w}^\star\leftarrow\mathbf{w}$.
\end{algorithmic}
\end{algorithm}

\begin{algorithm}[htbp]
\small
\caption{DLP--SMPC (Receding Horizon Execution) }
\label{alg:stochastic-mpc}
    \textbf{Input:} Prediction Horizon $H$, total simulation steps $T$, risk parameter $\gamma > 0$, weight bound $w_{\max}$,
    initial states $V_L(0), V_S(0)$, tolerance $\varepsilon$, max iterations $K_{\max}$,
    memory $m$, AL parameters $\rho>0$, $\lambda_0\geq 0$, \texttt{al\_maxiter}, \texttt{al\_tol}.

    \textbf{Notation:} $J_{\rho,\lambda}(\mathbf{w})$ is the augmented Lagrangian objective; $h(\mathbf{w}):=\mathbb{E}_k[V(k+H)]-V(k)$.
    \begin{algorithmic}[1]
    \FOR{$k = 0,1,2,\dots,T-1$}
        \STATE Update rolling estimates $\widehat\mu_k$, $\widehat\sigma_k^2$.
        \STATE Set $\mathbf{w}^{(0)}\in[0,w_{\max}]^H$, $\lambda\leftarrow\lambda_0$.
        \FOR{$\ell = 1,\dots,$ \texttt{al\_maxiter}}
            \STATE Solve inner subproblem via Algorithm~\ref{alg:lbfgsb} (implemented computationally by minimizing the negative objective~$-J_{\rho,\lambda}$):
            \[
                \mathbf{w}^\star \leftarrow \underset{\mathbf{w}\in[0,w_{\max}]^H}{\arg\max}\; J_{\rho,\lambda}(\mathbf{w})
            \]
            with initial point $\mathbf{w}^{(0)}$, memory $m$, tolerance $\varepsilon$, max iterations $K_{\max}$.
            \STATE Dual update: $\lambda\leftarrow\max(0,\,\lambda - \rho\,h(\mathbf{w}^\star))$.
            \IF{$\max(0,-h(\mathbf{w}^\star)) < $ \texttt{al\_tol}} \STATE \textbf{break} \ENDIF
            \STATE $\rho\leftarrow 2\rho$;\quad $\mathbf{w}^{(0)}\leftarrow\mathbf{w}^\star$.
        \ENDFOR
        \STATE Apply receding horizon: $w_k^*\leftarrow\mathbf{w}^\star_1$.
        \STATE $V_L(k+1)\leftarrow V_L(k)(1+w_k^*X(k))$,\\
                $V_S(k+1)\leftarrow V_S(k)(1-w_k^*X(k))$.
    \ENDFOR
    \STATE \textbf{return} $\{V_L(k),V_S(k),w_k^*\}_{k=0}^T$
    \end{algorithmic}
\end{algorithm}

\section{Empirical Illustrations}\label{section: empirical results}
This section presents illustrative examples that are backtested against historical data. Throughout this section, we initialize the long and short accounts symmetrically as~$V_L(0):= V_S(0) = \frac{1}{2}V(0)$ with $V(0):= \$100$. 
At each time step $k$, we estimate the sample mean and variance parameter pair ($\widehat\mu_k$, $\widehat\sigma_k^2$) using rolling sample statistics computed over a predefined window of the most recent $L > 1$ observations. That is,
$$
    \widehat{\mu}_k := \frac{1}{L}\sum_{i=1}^{L} X_{k-i}
\; \text{ and }  \;   
    \widehat{\sigma}_k^2 := \frac{1}{L-1}\sum_{i=1}^{L} (X_{k-i} - \widehat{\mu}_k)^2.
$$
In the empirical illustration below, these estimates are used uniformly across the prediction horizon, i.e., $\mu_i^{(k)} := \widehat{\mu}_k$ and $\bigl(\sigma_i^{(k)}\bigr)^2:= \widehat{\sigma}_k^2$ for $i = k, \dots, k+H-1$.
We evaluate the proposed DLP--SMPC approach, solved using L-BFGS-B, on daily closing prices of Bitcoin (Ticker: \texttt{BTC-USD}) obtained from Yahoo Finance. The sample period spans from~2019-12-31 to 2025-12-31, covering six years of trading data. This interval corresponds to broadly bullish, volatile price movement.  We compare the proposed DLP--SMPC approach against three classes of benchmarks: $(i)$ a constant-weight DLP benchmark with  $ w(k) \equiv  0.51 $ (selected via cross-validation), $(ii)$ a buy-and-hold strategy, and $(iii)$ three prescribed time-varying weight functions previously considered in \cite{wang2023robustness}. Letting~$N$ denote the total number of trading days in the sample period, these benchmark weight functions are 
\begin{align*}
    w_1(k) 
        &:= \log\left(1+\tfrac{k}{N} (e-1) \right) \\
    w_2(k) 
        &:= \tfrac{1}{2} \left( \sin \left( \tfrac{1}{\tfrac{0.02}{N}k - 0.01}\right)+1\right)\\
    w_3(k) 
        &:= f(k)\sin\left(\tfrac{1}{f(k)}\right)\mathds{1}_{\{f(k)\sin\left(\tfrac{1}{f(k)}\right)\ge 0\}}(k),
\end{align*}
with $f(k) := \left(\tfrac{4}{N}k - 2 \right)$.\footnote{
    According to \cite{wang2023robustness}, these three benchmark weighting functions serve as proxies for distinct investment philosophies. In particular, $w_1(\cdot)$ represents a monotonically increasing exposure, $w_2(\cdot)$ represents a highly active, oscillatory rebalancing strategy, and $w_3(\cdot)$ represents investing more at the beginning and end of the period, maintaining near-zero market exposure in the middle.  Note that to ensure strict mathematical well-posedness, any algebraic singularity occurring at exactly $k = N/2$ is resolved by taking the continuous limit of the respective functions.
}

Furthermore, we compare the proposed method with several alternative global optimization heuristics: Simulated Annealing~\cite{SA}, Differential Evolution~\cite{Storn1997DifferentialE}, and Basin Hopping~\cite{wales1997global}. All experiments use a training period of 2018-01-01 to 2019-12-31 and a held-out test period of 2020-01-01 to 2025-12-31. Hyperparameters for all methods are selected via cross-validation on the training period.\footnote{
    Hyperparameters for each global optimizer benchmark are selected via cross-validation where applicable. For Simulated Annealing, we use a geometric cooling schedule $T_{k+1} = \beta T_k$ with $\beta \in [0.90, 0.99]$, with initial temperature calibrated to yield an acceptance rate of $0.6$--$0.8$. For Basin Hopping, we use L-BFGS-B as the optimizer, with $1000$ iterations, an acceptance temperature $T = 5.0$, and a step size of $0.3$. For Differential Evolution, we set crossover probability $\mathrm{CR} \approx 0.3$, increasing to $[0.8, 1]$ if convergence stalls, and balance population size $\mathrm{NP}$ against mutation factor $F$ accordingly; see \cite{Storn1997DifferentialE} for parameter selection rule of thumb studied for Differential Evolution.
} The parameters $(\gamma, H, L)$ for DLP--SMPC are selected via uniform grid search over the candidate space $\gamma \in [0.1, 1]$, $H \in {2, 3, \dots, 50}$, and~$L \in {5, 6, \dots, 60}$. The triplet $(\gamma, H, L)$ reported in Figure~\ref{fig:btc_result} is the configuration that achieves the best cross-validation performance in terms of the mean-variance criterion.

\begin{figure}[htbp]
    \centering
    \includegraphics[width=0.9\linewidth]{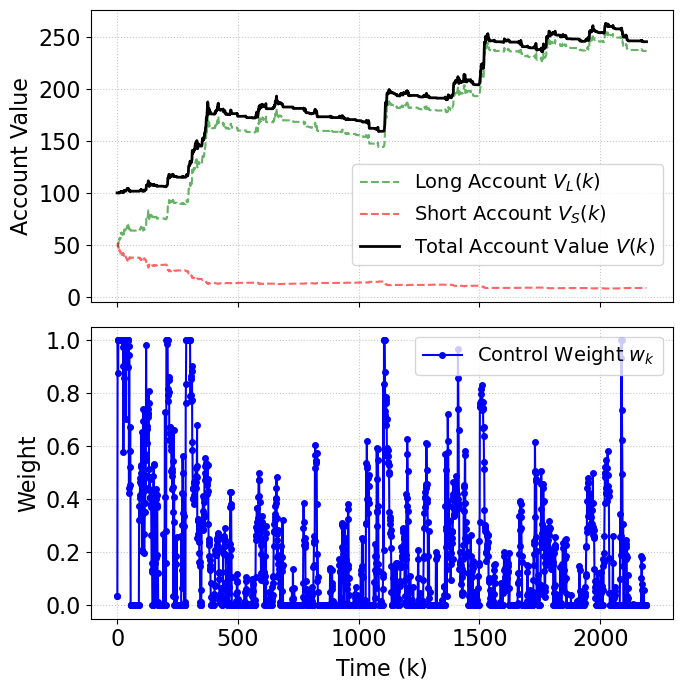}
    \caption[DLP--SMPC on BTC]{
   DLP--SMPC applied to the \texttt{BTC-USD} with cross-validated parameters, $\gamma = 0.1$, $H = 29$, and $L = 18$. 
   (Top): accounts value dynamics; 
   (Bottom) control-weight trajectory over time.}
    \label{fig:btc_result}
\end{figure}

Figure~\ref{fig:btc_result} shows that the proposed DLP--SMPC approach exhibits a step-like wealth trajectory, with a pattern of upward jumps separated by relatively flat periods. 
Table~\ref{performance} shows that DLP--SMPC achieves the highest annualized Sharpe ratio (1.388) and the lowest maximum drawdown (17.63\%) among all reported strategies. Its total return of 145.03\% is substantially lower than that of buy-and-hold (1129.29\%), but only marginally higher than the DLP with constant-weight~(139.68\%). This highlights a key distinction: while buy-and-hold achieves extreme returns, it does so at the cost of severe drawdowns (76.63\%), whereas the DLP with constant weight delivers neither competitive returns nor strong risk control (41.83\% drawdown). In contrast, DLP--SMPC attains an improved risk-adjusted profile, balancing moderate returns with substantially reduced downside risk.

Relative to the pre-defined time-varying strategies $w_i(k)$, DLP--SMPC consistently delivers superior risk-adjusted performance. Although $w_3(k)$ achieves a higher total return~(477.01\%), it incurs a much larger drawdown (46.05\%) and a lower Sharpe ratio (0.945), indicating an inferior risk--return trade-off. The remaining strategies, $w_1(k)$ and $w_2(k)$, are dominated by DLP--SMPC across both return and risk-adjusted metrics. When comparing global optimization methods within the SMPC framework, performance differences are relatively modest. Basin Hopping achieves the highest Sharpe ratio (1.388) and lowest drawdown (17.63\%), while Simulated Annealing and Differential Evolution yield slightly weaker but comparable results.

\begin{table}[htbp]
\footnotesize
\setlength{\tabcolsep}{4pt}
\renewcommand{\arraystretch}{0.85}
    \centering
    \caption{Trading Performance Summary}
    \label{performance}
    
    \newcolumntype{Y}{>{\raggedleft\arraybackslash}X}

    \begin{tabularx}{\linewidth}{l YYY}
    \toprule
    Metric & DLP-SMPC & DLP-Constant & Buy and Hold \\
    \midrule
    Total Return (\%)             
    & {145.03} & 139.68 & \textbf{1129.29} \\
    & {(130.17)} & (139.56) & \textbf{(1129.29)} \\

    Sharpe Ratio (Annualized)    
    & \textbf{1.39} & 0.85 & 0.996 \\
    & \textbf{(1.27)} & (0.85) & (0.996) \\

    Maximum Drawdown (\%)               
    & \textbf{17.63} & 41.83 & 76.63 \\
    & \textbf{(19.45)} & (41.83) & (76.63) \\

    Sortino Ratio (Annualized)  
    & \textbf{1.64} & 1.19 & 1.34 \\
    & \textbf{(1.59)} & (1.19) & (1.34) \\
    \bottomrule
    \end{tabularx}

    \begin{tabularx}{\linewidth}{l YYY}
    \toprule
    Metric & $w_1(\cdot)$ & $w_2(\cdot)$ & $w_3(\cdot)$ \\
    \midrule
    Total Return \%                 
    & 83.22 & 145.33 & {477.01} \\
    & (83.04) & (58.57) & {(474.53)} \\

    Sharpe Ratio (Annualized)       
    & 0.56 & 0.68 & {0.95} \\
    & (0.56) & (0.42) & {(0.94)} \\ 

    Maximum Drawdown \%                 
    & {30.05} & 33.86 & 46.05 \\ 
    & {(30.05)} & (39.12) & (46.06) \\

    Sortino Ratio (Annualized)  
    & 0.72 & 0.85 & {1.23} \\
    & (0.72) & (0.54) & {(1.24)} \\
    \bottomrule
    \end{tabularx}

\begin{tabularx}{\linewidth}{l YYY}
\toprule
Metric & Simulated Annealing & Differential Evolution & Basin Hopping \\
\midrule
Total Return (\%)              
& 115.37 & 139.29 & {145.03} \\
& (112.50) & (122.63) & {(129.26)} \\

Sharpe Ratio (Annualized)    
& 1.31 & 1.27 & {1.39} \\
& {(1.25)} & (1.14) & (1.24) \\

Maximum Drawdown (\%)               
& 18.10 & 18.02 & {17.63} \\
& (19.46) & (20.40) & {(19.46)} \\

Sortino Ratio (Annualized)  
& 1.56 & 1.67 & \textbf{1.67} \\
& (1.57) & (1.55) & {(1.58)} \\
\bottomrule
\end{tabularx}
    
    \smallskip
    {\footnotesize {Note: Parentheses denote results evaluated under the reduced-form control-adjusted cost model ($\varepsilon = 0.1\%$).}}
\end{table}

\subsection{Impact of Transaction Frictions}
We further investigate the DLP--SMPC formulation under a \emph{reduced-form} proxy for turnover costs, in which a proportional cost at rate $\varepsilon \in [0,1]$ is imposed on control adjustment at each rebalancing instance. 
Specifically, the cost is modeled as $\varepsilon \, | \Delta w_k | \, V_L(k)$ for the long account and $\varepsilon \, | \Delta w_k | \, V_S(k)$ for the short account, where $\Delta w_k:=  w_{k} - w_{k-1}$ denotes the change in the control weight.\footnote{ 
To guarantee one-step survivability under the reduced-form cost model, we redefine the maximum admissible weight as
$
w_{\max} := \min\!\left\{ \tfrac{1}{\varepsilon - X_{\min}},\, \tfrac{1}{X_{\max}+\varepsilon} \right\}.
$ This follows from the worst-case inequalities $1 \pm w_k X(k) -\varepsilon |\Delta w_k| >0$ together with the bound $|\Delta w_k| \leq w_{\max}.$ 
}Below, we set $\varepsilon = 0.10\%$ per trade, consistent with Binance's baseline spot trading fee scheduled for regular users; see \cite{binance_fee}.



Figure~\ref{fig:dlp_comparison} reveals a clear risk--return trade-off across all strategies. The buy-and-hold strategy remains fully exposed to BTC volatility, reaching \$1229 but suffering a drawdown exceeding~75\% around $k \approx 500$--$1000$, whereas DLP--SMPC produces a markedly smoother trajectory, terminating at approximately \$245 by dynamically reducing exposure during adverse periods. 

The predefined weight functions exhibit a similar divergence in performance: the aggressive strategy $w_3(k)$ achieves a higher terminal account value (approximately \$575) at the cost of large mid-sample drawdowns. In contrast, the DLP with constant weight terminates at a level comparable to DLP--SMPC but experiences significantly larger drawdowns. The remaining strategies, $w_1(k)$ and $w_2(k)$, underperform DLP--SMPC outright, terminating at lower values (approximately \$159--\$183). 
Collectively, these results confirm that the comparatively lower total return of DLP--SMPC is a direct consequence of trading upside capture for strict downside protection, which is a risk-adjusted balance that none of the predefined baselines replicate.

\begin{figure}[htbp]
    \centering
    \begin{subfigure}{0.95\linewidth}
        \centering
        \includegraphics[width=\linewidth]{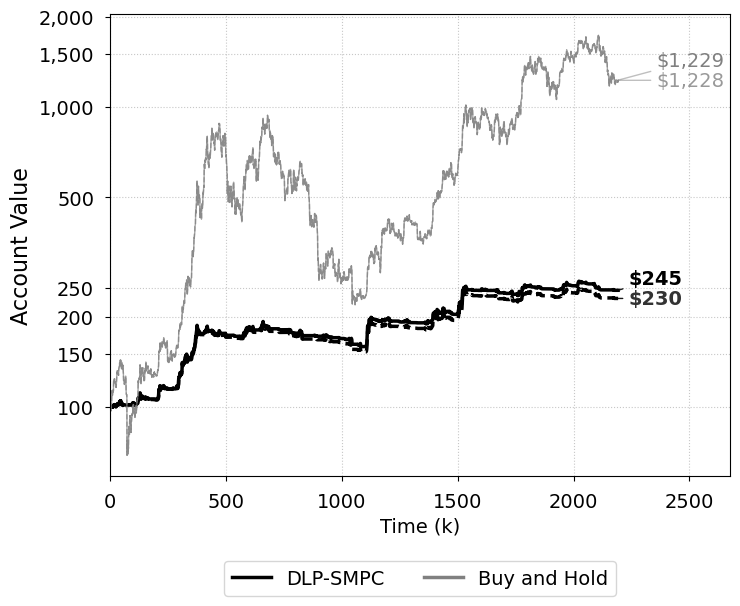}
        \caption{DLP--SMPC vs Buy-and-Hold}
    \end{subfigure}


    \begin{subfigure}{0.95\linewidth}
        \centering
        \includegraphics[width=\linewidth]{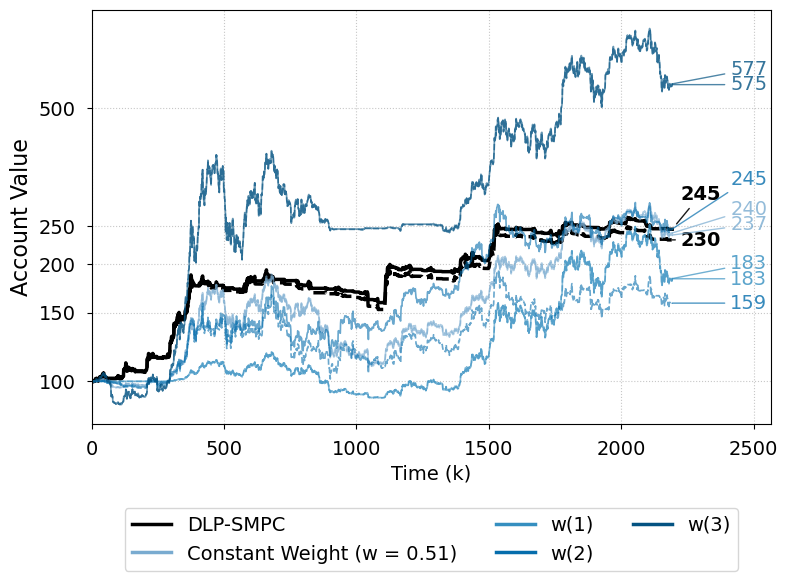}
        \caption{DLP--SMPC vs Different Weight Functions}
    \end{subfigure}


    \begin{subfigure}{0.95\linewidth}
        \centering
        \includegraphics[width=\linewidth]{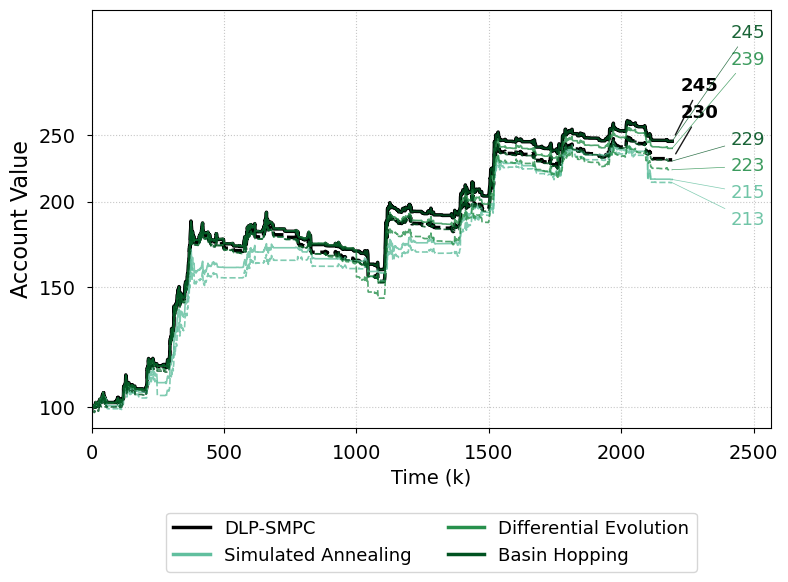}
        \caption{DLP--SMPC vs Global Optimization Algorithms}
    \end{subfigure}

    \caption{Performance comparison of the proposed DLP--SMPC with L-BFGS-B strategy against benchmark strategies. Dashed trajectories represent the integration of transaction cost $\varepsilon = 0.1\%$. Account values are displayed on a log scale.}
    \label{fig:dlp_comparison}
\end{figure}

Shifting to computational robustness, account value trajectories across all global optimizer variants are nearly indistinguishable over the entire sample period, with terminal values tightly clustered between \$213--\$245. 
Although not plotted here to conserve space, pointwise differences in weight trajectories between the algorithms are correspondingly small:
Basin Hopping deviates from L-BFGS-B by at most $\mathcal{O}(10^{-4})$, Differential Evolution exhibits persistent deviations on the order of $10^{-2}$, and Simulated Annealing shows larger initial deviations (up to $0.75$) before rapidly converging to near-zero differences. 
This close agreement between DLP--SMPC with L-BFGS-B and the global optimization methods suggests that, despite the non-convexity of the DLP--SMPC objective for $H \geq 2$, the proposed L-BFGS-B optimizer consistently identifies solutions of comparable quality to those obtained via global search, justifying its use as the primary solver without incurring the additional computational cost of global~methods.


\subsection{Cross-Asset Robustness Check}
To further assess robustness across asset classes beyond the Bitcoin study, we evaluate DLP--SMPC against constant-weight DLP and buy-and-hold on an additional set of assets, including Tesla (Ticker: \texttt{TSLA}), Ethereum (Ticker:~\texttt{ETH-USD}), and Apple Inc. (Ticker: \texttt{AAPL}). For each asset, we again report the performance metrics: Total Return, annualized Sharpe Ratio, Sortino Ratio, and Maximum Drawdown. The hyperparameters~$(\gamma, H, L)$ are selected via the same cross-validation protocol and data partition as in the Bitcoin experiment.

Across all three assets, DLP--SMPC consistently achieves the strongest risk-adjusted performance, outperforming both baselines on the Sharpe and Sortino ratios. This improvement is accompanied by a substantial reduction in downside risk: maximum drawdowns are contained within 10--13\% under DLP--SMPC, compared to 26--52\% for the DLP with constant weight and 33--79\% for buy-and-hold, representing a reduction of approximately 2--6$\times$ in peak-to-trough losses.

While buy-and-hold achieves the highest total returns across all assets (e.g., 1529.44\% for \texttt{TSLA} and 2192.57\% for \texttt{ETH-USD}), these gains are obtained at the cost of extreme volatility and severe drawdowns. The DLP with constant weight occupies an intermediate position: it reduces drawdown relative to buy-and-hold but still incurs substantial losses (up to 52.38\%) and delivers inferior risk-adjusted performance. In particular, its Sharpe and Sortino ratios are dominated by DLP--SMPC.
Overall, these results indicate that the performance of the proposed DLP--SMPC framework is consistent across asset classes, delivering downside protection while maintaining competitive returns.

\begin{table}[htbp]
    \footnotesize
    \setlength{\tabcolsep}{4pt}
    \renewcommand{\arraystretch}{0.9}
    \centering
    \caption{Performance across assets (parentheses denote transaction cost increment, $\varepsilon = 0.1\%$)}
    \label{tab:multi_asset}

\newcolumntype{Y}{>{\raggedleft\arraybackslash}X}

\begin{tabularx}{\linewidth}{l YYY }
\toprule
Metric & DLP-SMPC & DLP-Constant & Buy-and-Hold \\
\midrule

\multicolumn{4}{l}{\texttt{TSLA}} \\
Total Return (\%)      & 194.13 & 156.16 & \textbf{1529.44} \\
                        &(174.03)& (156.04) & (1529.44)\\
Sharpe Ratio (Annualized)          & \textbf{1.42} & 0.94 & 1.03  \\
                        &(1.314)        & (0.94) & (1.03)\\
Maximum Drawdown (\%)      & \textbf{12.17} & 37.37 & 73.63 \\
                        &(13.33)        & (37.37) &(73.63)\\
Sortino Ratio (Annualized)         & \textbf{2.52} & 1.16 & 1.58 \\
                        &(2.28)        &(1.16) & (1.58)\\

\midrule
\multicolumn{4}{l}{\texttt{ETH-USD}} \\
Total Return (\%)      & 249.84 & 327.55 & \textbf{2192.57} \\
                        & (231.33) & (327.32) & (2192.56) \\
Sharpe Ratio (Annualized)          & \textbf{1.70} & 0.84 & 1.05 \\
                        &(1.56)      &(0.84) & (1.05)\\
Maximum Drawdown (\%)      & \textbf{10.18} & 52.38 & 79.35 \\
                        &(11.51)         & (52.38)& (79.35)\\
Sortino Ratio (Annualized)         & \textbf{3.20} & 1.241 & 1.56 \\
                        &(2.90)         & (1.240) &  (1.56)    \\

\midrule
\multicolumn{4}{l}{\texttt{AAPL}} \\
Total Return (\%)      & 91.80 & 81.85 & \textbf{285.42} \\
                        &(73.91)& (81.69) & (285.42)\\
Sharpe Ratio (Annualized)          & \textbf{1.04} & 0.66 & 0.87 \\
                        &(0.87)        & (0.66) & (0.87)\\
Maximum Drawdown (\%)     & \textbf{12.85} & 26.38 & 33.36 \\
                        &(12.92)       & (26.38) & (33.36)\\
Sortino Ratio  (Annualized)        & \textbf{1.61} & 0.964 & 1.29 \\
                        &(1.34)        & (0.96) & (1.29)\\

\bottomrule
\end{tabularx}
    \smallskip
    {\footnotesize Selected hyperparameters $(\gamma,H,L)$: \texttt{ETH-USD} $(0.1, 11, 15)$; \texttt{TSLA} $(0.2, 12, 16)$; \texttt{AAPL} $(0.1, 17, 40)$.}
\end{table}

\section{Concluding Remarks}\label{conclusion}
In this paper, we propose an SMPC-based approach to dynamically optimize the weight selection~$w^*_k$ within the Double Linear Policy framework. Empirical evaluations using Bitcoin price data demonstrate that the proposed DLP--SMPC method improves risk-adjusted performance, particularly by constraining drawdowns and mitigating downside risk during the periods of high market volatility.

An interesting direction for future research is to relax the assumption that returns are $\mathcal{F}_k$-conditionally independent over the prediction horizon, to facilitate multi-asset portfolio settings. While conditional independence underlies our current SMPC formulation, extending this framework requires incorporating multivariate time-series models, such as Vector Autoregressive (VAR) or multivariate GARCH models, into the SMPC prediction horizon, allowing the control $w(k)$ to account for both serial dependence and cross-asset co-movements. Preliminary attempts to address this challenge have been made in \cite{hsieh2025robust}, where multi-asset risk management is studied in a more abstract market setting using generalized lattice-based~models.

\bibliographystyle{ieeetr}
\bibliography{refs}

\appendix 
 
\section{Technical Results} 
\label{appendix}
This section provides the derivation of the technical results.

\begin{proof}[Proof of Lemma~\ref{lemma: conditional moments}]
Since $\mathbf{z}_k$ is $\mathcal{F}_k$-measurable and the returns  $\{X(i)\}_{i=k}^{k+H-1}$ are conditionally independent given $\mathcal{F}_k$  with $\mathbb{E}_k[X(i)]=\mu_i^{(k)}$ and  $\mathbb{E}_k[X(i)^2] = \left(\mu_i^{(k)} \right)^2 + \left(\sigma_i^{(k)}\right)^2$, the system output satisfies
$
    y_{k+H} = \mathbf{c}^\top\Phi_{k+H,k}\,\mathbf{z}_k,   
$
where the state transition matrix is 
\begin{align*}
    \Phi_{k+H,k} 
    &=  A_{k+H-1} A_{k+H-2} \cdots A_k= \operatorname{diag}\!\left(d^+,\, d^-\right),
\end{align*}
with $d^\pm := \prod_{i=k}^{k+H-1}(1\pm w_iX(i)).$
Taking the conditional expectation $\mathbb{E}_k[\cdot]$, and using the conditional independence of $\{X(i)\}_{i=k}^{k+H-1}$ given $\mathcal{F}_k$ to factorize the expectation of the~products,
\begin{align*}
    & \mathbb{E}_k[\Phi_{k+H,k}] \\
    &= \operatorname{diag}\!\left(\prod_{i=k}^{k+H-1}\mathbb{E}_k[1+w_iX(i)],\;\prod_{i=k}^{k+H-1}\mathbb{E}_k[1-w_iX(i)]\right)\\ &= \operatorname{diag}(P^+,P^-) = \bar{\Phi},
\end{align*}
which yields $\mathbb{E}_k[y_{k+H}]=\mathbf{c}^\top\bar{\Phi}\,\mathbf{z}_k$.
Next, to compute the conditional variance, we first evaluate the second moment:
\[
y_{k+H}^2
=
(\mathbf c^\top \Phi_{k+H,k}\mathbf z_k)^2
=
\mathbf z_k^\top \Phi_{k+H,k}^\top \mathbf c \mathbf c^\top \Phi_{k+H,k}\mathbf z_k.
\]
Taking the conditional expectation gives
\[
\mathbb E_k[y_{k+H}^2]
=
\mathbf z_k^\top
\mathbb E_k\!\left[\Phi_{k+H,k}^\top \mathbf c \mathbf c^\top \Phi_{k+H,k}\right]
\mathbf z_k
=:
\mathbf z_k^\top M \mathbf z_k.
\]
Hence, the conditional variance evaluates to
\begin{align*}
    \operatorname{var}_k(y_{k+H})
    &=
    \mathbb E_k[y_{k+H}^2] - \bigl(\mathbb E_k[y_{k+H}]\bigr)^2\\
    &=
    \mathbf z_k^\top M \mathbf z_k
    -
    \mathbf z_k^\top \bar{\Phi}^\top \mathbf c \mathbf c^\top \bar{\Phi}\,\mathbf z_k\\
    &=
    \mathbf z_k^\top
    \bigl(M-\bar{\Phi}^\top \mathbf c \mathbf c^\top \bar{\Phi}\bigr)
    \mathbf z_k
    =
    \mathbf z_k^\top \Sigma\,\mathbf z_k.
\end{align*}
It remains to explicitly compute $M$. Since $\Phi_{k+H,k}=\operatorname{diag}(d^+,d^-)$ is diagonal, we note that
\[
\Phi_{k+H,k}^\top\mathbf{c}\mathbf{c}^\top\Phi_{k+H,k} 
= 
\begin{bmatrix} 
    (d^+)^2 & d^+d^- \\ d^+d^- & (d^-)^2 
\end{bmatrix}.
\]
Thus, $M_{pq}=\mathbb{E}_k[d^{(p)}d^{(q)}]$. Since $d^\pm=\prod_{i=k}^{k+H-1}(1\pm w_iX(i))$ and the $X(i)$ are conditionally independent, taking the conditional expectation and using independence to factorize each entry, we~have 
{\small
\begin{align*}
M_{11} 
    & = \mathbb{E}_k[(d^+)^2]
     = \prod_{i=k}^{k+H-1}\mathbb{E}_k\!\left[(1+w_iX(i))^2\right]\\ 
&= \prod_{i=k}^{k+H-1}\!\left(\left(1+w_i\mu_i^{(k)}\right)^2+\left(\sigma_i^{(k)}\right)^2w_i^2\right) = Q^+,\\
M_{22} 
    & = \mathbb{E}_k[(d^-)^2]
    = \prod_{i=k}^{k+H-1}\mathbb{E}_k\!\left[(1-w_iX(i))^2\right]\\ 
&= \prod_{i=k}^{k+H-1}\!\left(\left(1-w_i\mu_i^{(k)}\right)^2+\left(\sigma_i^{(k)}\right)^2w_i^2\right) = Q^-,\\
M_{12} 
    &= \mathbb{E}_k[d^+ d^-] 
    = \prod_{i=k}^{k+H-1}\mathbb{E}_k\!\left[1-w_i^2X(i)^2\right]
\\
&= \prod_{i=k}^{k+H-1}\!\left(1-w_i^2\left((\mu_i^{(k)})^2+\left(\sigma_i^{(k)}\right)^2\right)\right) = M_{21},
\end{align*}
}with  $\mathbb{E}_k[X(i)^2]= \left(\mu_i^{(k)} \right)^2 + \left(\sigma_i^{(k)} \right)^2$. Substituting $M$ into $\Sigma=M-\bar{\Phi}^\top\,\mathbf{c}\mathbf{c}^\top\bar{\Phi}$ completes the proof. 
\end{proof}

\begin{proof}[Proof of Theorem~\ref{theorem: gradient}] 
Fix $H>0$. We begin by analyzing the gradient for the conditional mean term $\mathbb{E}_k[y_{k+H}]  = \mathbf{c}^\top\bar{\Phi}\,\mathbf{z}_k$. From Lemma~\ref{lemma: conditional moments}, 
note that $\bar{\Phi}=\operatorname{diag}(P^+, P^-)$ with 
$P^\pm = \left(1 \pm \mu_{k+j}^{(k)} w_{k+j}\right) A_j^\pm$. Differentiating it with respect to $w_{k+_j}$ gives 
$
\frac{\partial\bar{\Phi}}{\partial w_{k+j}}= \mu_{k+j}^{(k)} D\bar{\Phi}_j,
$
where $D:= \operatorname{diag}(1, -1) $ and $\bar{\Phi}_j:= \operatorname{diag}(A_j^+, A_j^-)$.
Thus,
{\small
\begin{align} \label{eq: gradient of conditional }
\frac{\partial}{\partial w_{k+j}}\mathbb{E}_k[y_{k+H}] 
    &= 	\mathbf{c}^\top \frac{\partial\bar{\Phi}}{\partial w_{k+j}}\,\mathbf{z}_k
    = \mu_{k+j}^{(k)}\,\mathbf{c}^\top D\bar{\Phi}_j\,\mathbf{z}_k.
\end{align}
}
Next, we analyze the gradient for the variance term $\operatorname{var}_k(y_{k+H}) = \mathbf{z}_k^\top\Sigma\,\mathbf{z}_k$. Note that $\Sigma = M - \bar{\Phi}^\top\mathbf{c}\mathbf{c}^\top\bar{\Phi}$; hence differentiating it with respect to $w_{k+j}$ yields
{\small
\begin{align*}
    \frac{\partial\Sigma}{\partial w_{k+j}}     
        &= {M}'_j - \left(\frac{\partial\bar{\Phi}}{\partial w_{k+j}}\right)^\top \mathbf{c}\mathbf{c}^\top\bar{\Phi} - \bar{\Phi}\mathbf{c}\mathbf{c}^\top\frac{\partial\bar{\Phi}}{\partial w_{k+j}} \\ 
        &= {M}'_j - \mu_{k+j}^{(k)}\bigl(D\bar{\Phi}_j\mathbf{c}\mathbf{c}^\top\bar{\Phi} 
       + \bar{\Phi}\mathbf{c}\mathbf{c}^\top D\bar{\Phi}_j\bigr) =: \Sigma'_j.  
\end{align*}
}where last equality holds since $\bar{\Phi}$ and $\bar{\Phi}_j$ are diagonal, and hence $ \left(\frac{\partial \bar{\Phi}}{\partial w_{k+j}} \right)^\top = \frac{\partial \bar{\Phi}}{\partial w_{k+j}}$.

The entries of ${M}'_j$ follow from applying the product rule to the components of $M$ established in Lemma~\ref{lemma: conditional moments}.  For the diagonal entries $M_{11} = Q^+$ and $M_{22} = Q^-$, we have: 
{\small $$
    {m'}^\pm_j = \frac{\partial Q^\pm}{\partial w_{k+j}}=
    2\bigl(\pm\mu_{k+j}^{(k)}\bigl(1\pm\mu_{k+j}^{(k)} w_{k+j}\bigr)
    +\bigl(\sigma_{k+j}^{(k)}\bigr)^2 w_{k+j}\bigr)B_j^\pm
$$
}For the off-diagonal entries, we have
$\displaystyle \frac{\partial M_{12}}{\partial w_{k+j}}
= -2 \bigl((\mu_{k+j}^{(k)})^2 + (\sigma_{k+j}^{(k)})^2\bigr) w_{k+j} C_j$,
using
$\displaystyle \frac{\partial}{\partial w_{k+j}} \bigl( 1 - w_{k+j}^2 ((\mu_{k+j}^{(k)})^2 + (\sigma_{k+j}^{(k)})^2) \bigr)
= -2 ((\mu_{k+j}^{(k)})^2 + (\sigma_{k+j}^{(k)})^2) w_{k+j}$.
Combining the derivatives of the mean and variance terms~yields
\begin{align*}
\frac{\partial }{\partial w_{k+j}}J(\mathbf{w}) =
\mu_{k+j}^{(k)}\,\mathbf{c}^\top D\bar{\Phi}_j\mathbf{z}_k 
- \gamma\,\mathbf{z}_k^\top\Sigma'_j\mathbf{z}_k.  \qquad \qedhere    
\end{align*}

\end{proof}

\end{document}